\documentclass[letterpaper, 10 pt, conference]{ieeeconf}  

\IEEEoverridecommandlockouts                              
\usepackage{amsmath,amssymb,amsfonts, mathtools}
\usepackage{theorem}
\usepackage{siunitx}
\usepackage{bm}
\usepackage{graphicx}
\usepackage{url}
\usepackage{algorithm}
\usepackage{algcompatible}
\usepackage{hhline}
\usepackage{mathrsfs}
\usepackage{pgfplots}
\usepgfplotslibrary{fillbetween}

\graphicspath{{pics/}}

\newcommand{\mat}[1]{\ensuremath{\begin{bmatrix} #1 \end{bmatrix}}}
\newcommand{\vc}[1]{\ensuremath{\begin{bmatrix} #1 \end{bmatrix}}}

\newcommand{\ul}[1]{\underline{#1}}
\newcommand{\eps}{\varepsilon}
\renewcommand{\rho}{\varrho}
\renewcommand{\Pr}[1]{\mathrm{Pr}\left[#1\right]}
\newcommand{\E}[1]{\mathrm{E}\left[#1\right]}

\newcommand{\norm}[1]{\left\lVert #1 \right\rVert}

\renewcommand{\text}[1]{\mathrm{#1}}
\newcommand{\conv}[1]{\mathrm{conv}\left(#1\right)}
\newcommand{\col}[1]{\mathrm{col}\left(#1\right)}
\newcommand{\lag}[1]{\mathrm{lag}\left(#1\right)}

\newcommand{\conf}{\beta}

\newtheorem{theorem}{Theorem}

\newtheorem{proposition}{Proposition}
\newtheorem{remark}{Remark}
\newtheorem{corollary}{Corollary}
\newtheorem{definition}{Definition}
\newtheorem{assumption}{Assumption}

\usepackage[acronym, style=alttree, toc=true, shortcuts, xindy, nomain, nonumberlist]{glossaries}

\newglossary{symbols}{sym}{sbl}{List of Symbols}
\newglossary{notation}{not}{nt}{Notation}

\glsaddkey{dimension}{\glsentrytext{\glslabel}}{\glsentryd}{\GLsentryd}{\glsd}{\Glsd}{\GLSd}

\newglossaryentry{x}{type=symbols,
	sort={x},
	dimension={\ensuremath{ n }},
	name={\ensuremath{\bm{x}}},
	description={State}
}

\newglossaryentry{xi}{type=symbols,
	sort={xi},
	dimension={\ensuremath{ n_{\xi} }},
	name={\ensuremath{\bm{\xi}}},
	description={Extended state}
}

\newglossaryentry{y}{type=symbols,
	sort={y},
	dimension={\ensuremath{ p }},
	name={\ensuremath{\bm{y}}},
	description={Output}
}

\newglossaryentry{u}{type=symbols,
	sort={u},
	dimension={\ensuremath{ m }},
	name={\ensuremath{\bm{u}}},
	description={Input}
}

\newglossaryentry{v}{type=symbols,
	sort={v},
	dimension={\ensuremath{ m }},
	name={\ensuremath{\bm{v}}},
	description={Input correction term}
}

\newglossaryentry{d}{type=symbols,
	sort={d},
	dimension={\glsd{y}},
	name={\ensuremath{\bm{d}}},
	description={Disturbance}
}

\newglossaryentry{dset}{type=symbols,
	sort={dset},
	dimension={\glsd{y}},
	name={\ensuremath{\mathbb{D}}},
	description={Disturbance constraint set}
}

\newglossaryentry{ud}{type=symbols,
	sort={ud},
	dimension={\ensuremath{ \mathcal{U} }},
	name={\ensuremath{\bm{u}^{\mathrm{d}}}},
	description={Input data}
}

\newglossaryentry{vd}{type=symbols,
	sort={vd},
	dimension={\ensuremath{ \mathcal{V} }},
	name={\ensuremath{\bm{v}^{\mathrm{d}}}},
	description={Input correction term data}
}

\newglossaryentry{dd}{type=symbols,
	sort={dd},
	dimension={\ensuremath{ \mathcal{D} }},
	name={\ensuremath{\bm{d}^{\mathrm{d}}}},
	description={Disturbance data}
}

\newglossaryentry{xid}{type=symbols,
	sort={xid},
	dimension={\ensuremath{ \mathcal{X} }},
	name={\ensuremath{\bm{\xi}^{\mathrm{d}}}},
	description={Extended state data}
}

\newglossaryentry{yd}{type=symbols,
	sort={yd},
	dimension={\ensuremath{ \mathcal{Y} }},
	name={\ensuremath{\bm{y}^{\mathrm{d}}}},
	description={Output data}
}

\newglossaryentry{Tini}{type=symbols,
	sort={Tini},
	dimension={\ensuremath{ 1 }},
	name={\ensuremath{T_{\mathrm{p}}}},
	description={Upper bound of lag}
}

\newglossaryentry{Tf}{type=symbols,
	sort={Tf},
	dimension={\ensuremath{ 1 }},
	name={\ensuremath{T_{\mathrm{f}}}},
	description={Prediction horizon}
}

\newglossaryentry{yref}{type=symbols,
	sort={y},
	dimension={\ensuremath{ p }},
	name={\ensuremath{\bm{y}^{\mathrm{ref}}}},
	description={Output}
}

\newglossaryentry{uref}{type=symbols,
	sort={u},
	dimension={\ensuremath{ m }},
	name={\ensuremath{\bm{u}^{\mathrm{ref}}}},
	description={Input}
}

\newglossaryentry{yrefvec}{type=symbols,
	sort={y},
	dimension={\ensuremath{ p }},
	name={\ensuremath{\ul{\bm{y}}^{\mathrm{ref}}}},
	description={Output}
}

\newglossaryentry{urefvec}{type=symbols,
	sort={u},
	dimension={\ensuremath{ m }},
	name={\ensuremath{\ul{\bm{u}}^{\mathrm{ref}}}},
	description={Input}
}

\graphicspath{{pics/}}

\title{\LARGE \bf
Adaptive Stochastic Predictive Control from Noisy Data: \\ A Sampling-based Approach
}

\author{Johannes Teutsch, Christopher Narr, Sebastian Kerz, Dirk Wollherr, and Marion Leibold
\thanks{All authors are with the Chair of Automatic Control Engineering (LSR), Department of Computer Engineering,
        Technical University of Munich, Theresienstr. 90, 80333 Munich, Germany
        {\tt\small \{johannes.teutsch, c.narr, s.kerz, dirk.wollherr, marion.leibold\}@tum.de}}%
}

\newcommand\copyrighttext{%
  \footnotesize \textcopyright 2024 IEEE. Personal use of this material is permitted.
  Permission from IEEE must be obtained for all other uses, in any current or future
  media, including reprinting/republishing this material for advertising or promotional
  purposes, creating new collective works, for resale or redistribution to servers or
  lists, or reuse of any copyrighted component of this work in other works.}
\newcommand\copyrightnotice{%
\begin{tikzpicture}[remember picture,overlay]
\node[anchor=south,yshift=10pt] at (current page.south) {\fbox{\parbox{\dimexpr\textwidth-\fboxsep-\fboxrule\relax}{\copyrighttext}}};
\end{tikzpicture}%
}

\begin{document}

\maketitle
\copyrightnotice

\thispagestyle{empty}
\pagestyle{empty}

\begin{abstract}
    In this work, an adaptive predictive control scheme for linear systems with unknown parameters and bounded additive disturbances is proposed. In contrast to related adaptive control approaches that robustly consider the parametric uncertainty, the proposed method handles all uncertainties stochastically by employing an online adaptive sampling-based approximation of chance constraints. The approach requires initial data in the form of a short input-output trajectory and distributional knowledge of the disturbances. This prior knowledge is used to construct an initial set of data-consistent system parameters and a distribution that allows for sample generation. As new data stream in online, the set of consistent system parameters is adapted by exploiting set membership identification. Consequently, chance constraints are deterministically approximated using a probabilistic scaling approach by sampling from the set of system parameters. In combination with a robust constraint on the first predicted step, recursive feasibility of the proposed predictive controller and closed-loop constraint satisfaction are guaranteed. A numerical example demonstrates the efficacy of the proposed method.

\end{abstract}

\section{Introduction}
Model predictive control (MPC) is widely regarded as the state-of-the-art method in performance-oriented control of constrained dynamical systems. As modern systems become more complex and sensors more abundant, designing predictive controllers from measurement data is becoming increasingly appealing.
For linear time-invariant (LTI) systems, predictive controllers may be constructed directly from data by representing all possible future finite-time input-output trajectories as linear combinations of available input-output data trajectories \cite{willems2005note,behavioraltheory2021,berberich2020data}. If the LTI system is subject to disturbances, stochastic predictive control promises efficient control with user-specified chance constraints that enable a trade-off between performance and safety.  
Stochastic data-driven predictive control schemes as in \cite{pan2023data,kerz2023datadriven} come with guarantees of recursive feasibility and chance constraint satisfaction, but rely on exact data in the sense that an input-output-\emph{disturbance} trajectory needs to be available. 
If the disturbance trajectory is unknown, considering all possible disturbance sequences that are consistent with the available input-output data, akin to a set membership identification (SMI) perspective, leads to a family of data-consistent predictors from which samples enable control design \cite{teutsch2024sampling}.

However, there is still a conceptual drawback inherent to the approaches above, in which data are treated as something that is available only offline to facilitate design. Online, during the control phase, outputs are repeatedly measured, used for control, and then discarded. In the presence of process noise, it is not immediately obvious how online data may be incorporated to improve control performance of the above-mentioned data-driven control schemes. 
In contrast, SMI suggests that online improvements are possible in principle, as the set of consistent system parameters shrinks and converges to a singleton given tight disturbance bounds and sufficiently rich inputs \cite{bai1998ConvergenceSMI}. 
The idea of SMI is not novel in the data-driven literature, e.g., it is used in the robust design of controllers from noisy data \cite{berberich2020robustfb} and implicit in the data-informativity framework \cite{datainformativity}. However, SMI has not yet been exploited for online updates of constraints and predictors in data-driven predictive control. 

In the MPC literature, the idea of updating a set of model parameters and adapting constraints or predictors appropriately goes far back \cite{veres1993smipc,tanaskovic2014adaptive} and is often referred to as adaptive MPC. For LTI systems subject to parametric model uncertainty, \cite{lorenzen2019robust} presents an adaptive robust MPC scheme. By integrating the set of consistent system parameters and the adaptation of tubes into one single optimization problem, the computational load of online updates is kept small.
In \cite{lu2023adaptivePE}, an additional condition on the predicted input sequence in the optimal control problem (OCP) guarantees persistent excitation and, thereby, convergence to the true system parameters.
In the presence of chance constraints and probabilistic disturbances, \cite{bujarbaruah2020adaptive} presents an adaptive stochastic MPC scheme for systems subject to a bounded time-varying offset --- a setting related to parametric model uncertainty --- where the domain of the unknown offset is iteratively refined online and treated robustly.
In \cite{arcari2023stochastic}, an adaptive stochastic MPC scheme is proposed for a setting with parametric model uncertainty. The parametric uncertainty is treated robustly by constructing homothetic tubes along the prediction horizon, while the additive disturbance is treated probabilistically based on robustified probabilistic reachable sets. Nevertheless, the robust treatment of parametric uncertainty might result in overly conservative control. 

As stochastic handling of parametric uncertainty is generally challenging due to complex uncertainty propagation, sampling-based approaches offer a simple remedy. A sampling-based predictive control scheme is presented in \cite{mammarella2018offline}, based on the offline uncertainty sampling approach of \cite{lorenzen2017stochastic} for the deterministic reformulation of the chance constraints. In \cite{Mammarella2020efficient}, the computational efficiency was improved by considering inner-approximations of the chance constraint set found through scaling of fixed complexity sets. However, possible adaptation of the controller using newly obtained measurement data during control is not considered.

\subsubsection*{Contribution}
In this work, we develop an adaptive stochastic predictive control algorithm for LTI systems with unknown system parameters subject to bounded additive disturbances. 
Based on input-output trajectory data, the key idea is to map the probability distribution and bounds of the unknown disturbances to a probability distribution over the system parameters that are consistent with the available data. This allows for deterministic approximation of the imposed chance constraints on the output with high confidence by employing a sampling-based probabilistic scaling approach \cite{mammarella2022chance} using samples of consistent system parameters and possible future disturbances. Online, as new input-output data become available, we use SMI to iteratively adapt the support of the probability distribution over consistent system parameters for sample generation. This leads to a reduction of model uncertainty and thus allows for increased control performance during closed-loop operation by rescaling the previously approximated chance constraints. A robust constraint on the first predicted step enables guarantees of recursive feasibility and closed-loop constraint satisfaction despite the adaptation scheme. In contrast to related adaptive control approaches (e.g., \cite{arcari2023stochastic}), the proposed sampling-based approach handles \textit{all} uncertainties stochastically.

\subsubsection*{Organization}
In Section~\ref{sec:setup}, we introduce the problem setup and preliminaries on sampling-based approximation of chance constraints. The design steps for the proposed controller and its theoretical guarantees are given in Section~\ref{sec:method}. Section~\ref{sec:eval} provides a numerical evaluation of the proposed controller, before we conclude the work in Section~\ref{sec:conclusion}.

\subsubsection*{Notation}
We write $\bm{0}_{m\times n}$ for a zero matrix of dimension $m \times n$ and $\bm{I}_n$ for the identity matrix of dimension $n \times n$. When the dimension is clear from the context, we omit the index. With $\bm{1}_n \in \mathbb{R}^n$, we denote a column-vector of all ones. We abbreviate the set of integers $\left\{a,\,\ldots,\, b\right\}$ by $\mathbb{N}_a^b$. The Moore-Penrose pseudo-rightinverse of a matrix $\bm{S}$ is defined as $\bm{S}^{\dagger} \coloneqq \bm{S}^{\top}\left(\bm{S}\bm{S}^{\top}\right)^{-1}$. The probability measure and the expectation operator conditioned on the measurements at time step $k$ are denoted by $\Pr{\cdot}$ and $\E{\cdot}$, respectively. With $\bm{S}_1 \otimes \bm{S}_2$, we denote the Kronecker product of the matrices $\bm{S}_1$, $\bm{S}_2$. By $\col{\bm{s}_a,\,\ldots,\,\bm{s}_b} \coloneqq \mat{\bm{s}^{\top}_a\,\cdots\,\bm{s}^{\top}_b}^{\top}$, we denote the result from stacking the vectors/matrices $\bm{s}_a,\,\ldots,\,\bm{s}_b$. For a matrix $\bm{S}$, we define the weighted 2-norm of the vector $\bm{s}$ as $\norm{\bm{s}}_{\bm{S}} \coloneqq \sqrt{\bm{s}^{\top} \bm{S} \bm{s}}$. We write $\bm{y}_{i|k}$ for the predicted output $i$ steps ahead of time step $k$. For any sets $\mathbb{S}_1,\mathbb{S}_2$, we write the Minkowski set addition as $\mathbb{S}_1 \oplus \mathbb{S}_2=\{\bm{s}_1 + \bm{s}_2 \mid \bm{s}_1 \in \mathbb{S}_1,~ \bm{s}_2 \in \mathbb{S}_2\}$ and the Pontryagin set difference as $\mathbb{S}_1 \ominus \mathbb{S}_2 = \{\bm{s}_1 \in \mathbb{S}_1 \mid \bm{s}_1 + \bm{s}_2 \in \mathbb{S}_1, ~\forall  \bm{s}_2 \in \mathbb{S}_2 \}$. Positive definiteness is denoted by $\bm{S} \succ \bm{0}$. We denote the ceil-function as $\lceil\cdot\rceil$, and we write $\left[\bm{s}\right]_{a}$ for the $a$-th element of the vector $\bm{s}$. Lastly, the operator ``$\le$" applies element-wise. 

\section{Problem Setup and Preliminaries} \label{sec:setup}
In this section, we present the problem setup and introduce the sampling-based approach for deterministic approximation of chance constraints as given in~\cite{mammarella2022chance}.

\subsection{Problem Setup}
We consider a discrete-time LTI system of the form
\begin{subequations} \label{eq:system_minimal}
\begin{align}
	\gls{x}_{k+1} &= \bm{A} \gls{x}_{k} + \bm{B} \gls{u}_{k} + \bm{E}\gls{d}_{k}, \\
	\gls{y}_k &= \bm{C} \gls{x}_{k} + \bm{D} \gls{u}_{k} + \gls{d}_{k},
\end{align}
\end{subequations}
with state $\gls{x}_k \in \mathbb{R}^{\glsd{x}}$, output $\gls{y}_{k}\in \mathbb{R}^{\glsd{y}}$, input $\gls{u}_{k} \in \mathbb{R}^{\glsd{u}}$, additive disturbance $\gls{d}_{k} \in \mathbb{R}^{\glsd{d}}$. The system matrices $\bm{A}$, $\bm{B}$, $\bm{C}$, $\bm{D}$, $\bm{E}$ are unknown, and measurements of the state $\gls{x}_{k}$ and disturbance $\gls{d}_{k}$ cannot be obtained. However, system~\eqref{eq:system_minimal} satisfies the following assumption.
\begin{assumption}[System properties]\label{assum:properties}
    System \eqref{eq:system_minimal} is controllable and observable, and an upper bound $\gls{Tini} \ge \lag{\bm{A},\bm{C}}$ is known, with $\lag{\bm{A},\bm{C}}$ being the smallest natural number $j \le n$ for which $\col{\bm{C}, \bm{C} \bm{A}, \ldots, \bm{C} \bm{A}^{j-1}}$ has rank $n$. Moreover, the input-output behavior of system~\eqref{eq:system_minimal} is representable by the stabilizable and detectable dynamics 
    \begin{subequations}\label{eq:system_arx_nonminimal}
    \begin{align}
        \gls{xi}_{k+1} &= \tilde{\bm{A}} \gls{xi}_{k} + \tilde{\bm{B}} \gls{u}_{k} + \tilde{\bm{E}}\gls{d}_{k},\\
        \gls{y}_{k} &= \bm{\Phi} \gls{xi}_{k} + \bm{\Psi} \gls{u}_{k} + \gls{d}_{k}, \label{eq:system}
    \end{align}
    \end{subequations}
    with the (not necessarily minimal) extended state $\gls{xi}_{k}$ given~as
    \begin{equation}
        \gls{xi}_{k} \coloneqq \vc{\col{\gls{u}_{k-\gls{Tini}},\,\dots,\,\gls{u}_{k-1}} \\ \col{\gls{y}_{k-\gls{Tini}},\,\dots,\,\gls{y}_{k-1}}} \in \mathbb{R}^{\glsd{xi}},\,\glsd{xi}\coloneqq(m+p)\gls{Tini},\label{eq:extstate}
    \end{equation}
    and the system matrices $\tilde{\bm{A}} \coloneqq \col{\bar{\bm{A}},\,\bm{\Phi}}$, $\tilde{\bm{B}} \coloneqq \col{\bar{\bm{B}},\,\bm{\Psi}}$, $\tilde{\bm{E}} \coloneqq \col{\bm{0},\, \bm{I}_{\glsd{y}}}$ for some $\bm{\Phi}$ and $\bm{\Psi}$, where
    \begin{equation}
        \bar{\bm{A}} \coloneqq \mat{\bm{0} & \bm{I}_{(\gls{Tini}-1)\glsd{u}} & \bm{0} & \bm{0} \\ \bm{0} & \bm{0} & \bm{0} & \bm{0} \\ \bm{0} & \bm{0} & \bm{0} & \bm{I}_{(\gls{Tini}-1)\glsd{y}} },\, \bar{\bm{B}} \coloneqq \mat{\bm{0}\\ \bm{I}_{\glsd{u}} \\ \bm{0}}. \label{eq:system_arx_nonminimal_helper}
    \end{equation}
\end{assumption}

System \eqref{eq:system_minimal} is subject to probabilistic output and hard input constraints for all time steps $k \ge 0$, given as
\begin{subequations}\label{eq:constraints}
\begin{align}
	&\Pr{\gls{y}_{k} \in \mathbb{Y}} \ge 1-\eps,~ \mathbb{Y} = \left\{\bm{y} \in \mathbb{R}^{\glsd{y}}  ~\left|~ \bm{G}_y \,\bm{y} \le \bm{g}_y \right.\right\}, \label{eq:outputcons} \\
	&\hspace{5mm}\gls{u}_{k} \in  \mathbb{U}, \hspace{14.6
	mm} \mathbb{U} = \left\{\bm{u} \in \mathbb{R}^{\glsd{u}} \left|~ \bm{G}_u \bm{u} \le \bm{g}_u \right.\right\}, \label{eq:inputcons}
	\end{align}
\end{subequations}
where $\mathbb{Y}$ and $\mathbb{U}$ are compact sets containing the origin, and $\eps \in (0,\,1)$ is the risk parameter. Moreover, the disturbance $\gls{d}_{k}$ satisfies the following assumption. 
\begin{assumption}[Disturbance bounds and distribution] \label{assum:distset}
    The disturbance $\gls{d}$ is the realization of a zero-mean random variable that is independent and identically distributed~(iid) according to a known probability density function $\bm{f}_{\gls{d}}(\cdot)$, with known compact polytopic support set
\begin{equation}\label{eq:distset}
    \gls{dset} = \left\{ \gls{d} \in \mathbb{R}^{\glsd{d}} ~\left|~ \bm{G}_d \bm{d} \le \bm{g}_d \right.\right\}.
\end{equation}
\end{assumption}

This paper aims to design an adaptive output-feedback predictive control scheme for system \eqref{eq:system_minimal}.
Given user-specified weighting matrices $\bm{Q}$, $\bm{R} \succ \bm{0}$, 
tracking references $\gls{yref}_k$, $\gls{uref}_k$ for all $ k \ge 0$, and prediction horizon $\gls{Tf} \ge 1$, the proposed controller aims to minimize the expected finite horizon cost
\begin{equation}\label{eq:cost_expected}
    J \coloneqq \E{\sum\limits_{l=0}^{\gls{Tf}-1} \left( \norm{\gls{y}_{l|k} - \gls{yref}_{k+l}}^2_{\bm{Q}} + \norm{\gls{u}_{l|k} - \gls{uref}_{k+l}}^2_{\bm{R}} \right)}
\end{equation}
in a receding horizon fashion while guaranteeing satisfaction of output chance constraints~\eqref{eq:outputcons} and input constraints \eqref{eq:inputcons}. Note that terminal costs and constraints are intentionally left out, as a stability analysis is out of the scope of this work.

Since the system matrices in \eqref{eq:system_minimal} and \eqref{eq:system_arx_nonminimal} are unknown in our problem setting, we assume to instead have access to an initial input-output data trajectory, collected offline before the control phase. This assumption is specified in more detail in Section~\ref{sec:method}. Furthermore, newly retrieved input-output data during control is used to adapt the controller.

For tractability of the proposed predictive control scheme, the chance constraints \eqref{eq:outputcons} need to be reformulated into a deterministic expression. 
As the true system parameters are unknown, an exact reformulation is impossible.
In Section~\ref{sec:method}, the input-output data, disturbance bounds, and disturbance distribution will yield a probability distribution over system parameters with bounded support. In turn, this enables a sampling-based probabilistic scaling approach \cite{mammarella2022chance} for an inner-approximation of the true chance constraint set.
Since the resulting approximating sets can be easily rescaled online, this approach allows for an efficient online adaptation of the chance constraint approximation during control. The following subsection briefly presents this probabilistic scaling approach; for a broader discussion, we refer to~\cite{mammarella2022chance}.

\subsection{Probabilistic Scaling} \label{sec:sampling}
Consider a general joint chance constraint $\Pr{\bm{G}_{\zeta}(\bm{w}) \bm{\zeta} \le \bm{g}_{\zeta}(\bm{w})} \ge 1-\eps$ where $\bm{\zeta}$ is the (deterministic) decision variable, and $\bm{G}_{\zeta}(\bm{w})\in\mathbb{R}^{n_{\mathrm{c}} \times n_{\zeta}}$, $\bm{g}_{\zeta}(\bm{w}) \in \mathbb{R}^{n_{\mathrm{c}}}$ are constraint parameters that depend on the realization $\bm{w} \in \mathbb{R}^{n_w}$ of a multivariate random variable. The corresponding $\eps$-chance constraint set ($\eps$-CSS) is defined as 
\begin{equation} \label{eq:chanceconsset}
    \mathbb{Z}^{\mathrm{P}} = \left\{ \bm{\zeta} \in \mathbb{R}^{n_{\zeta}} ~\left|~ \Pr{\bm{G}_{\zeta}(\bm{w}) \bm{\zeta} \le \bm{g}_{\zeta}(\bm{w})} \ge 1-\eps \right. \right\}.
\end{equation}
The goal of the probabilistic scaling approach is to determine a deterministic inner-approximation $\mathbb{Z}^{\mathrm{S}}$ of the $\eps$-CCS \eqref{eq:chanceconsset} by using $N_{\mathrm{S}}$ iid uncertainty samples $\bm{w}^{(i)}$, $i \in \mathbb{N}_1^{N_{\mathrm{S}}}$, yielding $\Pr{\mathbb{Z}^{\mathrm{S}} \subseteq \mathbb{Z}^{\mathrm{P}}} \ge 1-\conf$ with a desired, user-chosen level of confidence $\conf$. The idea of the approach is to approximate the $\eps$-CSS \eqref{eq:chanceconsset} by probabilistically scaling a pre-defined scalable simple approximating set (SAS) of fixed complexity, i.e., 
\begin{equation}
    \mathbb{Z}^{\mathrm{S}}(\sigma) \coloneqq \{\bm{\zeta}_{\mathrm{c}}\} \oplus \sigma \mathbb{Z}^{\mathrm{SAS}},
\end{equation}
with center $\bm{\zeta}_{\mathrm{c}}$, shape $\mathbb{Z}^{\mathrm{SAS}}$, and scaling factor $\sigma \ge 0$. The scaling factor $\sigma$ of the SAS $\mathbb{Z}^{\mathrm{S}}$ for a single sample $\bar{\bm{w}}$ of the uncertainty $\bm{w}$ is specified as follows. 
\begin{definition}[Scaling Factor~\cite{mammarella2022chance}] \label{def:scaling}
    For a given sample $\bar{\bm{w}}$ and an SAS $\mathbb{Z}^{\mathrm{S}}(\sigma)$ with center $\bm{\zeta}_{\mathrm{c}}$ and shape $\mathbb{Z}^{\mathrm{SAS}}$, the scaling factor $\sigma(\bar{\bm{w}}) \ge 0$ of $\mathbb{Z}^{\mathrm{S}}(\sigma)$ relative to $\bar{\bm{w}}$ is defined~as
    \begin{equation} \label{eq:scaling}
        \sigma(\bar{\bm{w}}) \coloneqq \left\{ \begin{array}{cc}
            \max\limits_{\mathbb{Z}^{\mathrm{S}}(\sigma) \subseteq \mathbb{Z}(\bar{\bm{w}})} \sigma & \mathrm{if }~ \bm{\zeta}_{\mathrm{c}} \in \mathbb{Z}(\bar{\bm{w}}) \\
            0 & \mathrm{otherwise,}
        \end{array}\right.
    \end{equation}
    with $\mathbb{Z}(\bar{\bm{w}}) \coloneqq \left\{ \bm{\zeta} \in \mathbb{R}^{n_{\zeta}} ~\left|~ \bm{G}_{\zeta}(\bar{\bm{w}}) \bm{\zeta} \le \bm{g}_{\zeta}(\bar{\bm{w}}) \right. \right\}$.
\end{definition}
The following proposition provides means of obtaining an optimal scaling factor $\sigma^*$ such that $\Pr{\mathbb{Z}^{\mathrm{S}}(\sigma^*) \subseteq \mathbb{Z}^{\mathrm{P}}} \ge 1-\conf$ with confidence $\conf$, based on samples of the uncertainty~$\bm{w}$.
\begin{proposition}[Probabilistic Scaling of SAS~\cite{mammarella2022chance}] \label{prop:scaling}
    For a given candidate SAS $\mathbb{Z}^{\mathrm{S}}(\sigma)$ with center $\bm{\zeta}_{\mathrm{c}} \in \mathbb{Z}^{\mathrm{P}}$, risk parameter $\eps \in (0,1)$, and confidence $\conf \in (0,1)$, let the sample complexity $N_{\mathrm{S}}$ be chosen as $N_{\mathrm{S}} \ge N_{\mathrm{PS}}(\eps,\,\conf)$, with
    \begin{equation} \label{eq:samplecomplexity_scaling}
        N_{\mathrm{PS}}(\eps,\,\conf) \coloneqq \frac{\left(1+\sqrt{3}\right)^2}{\eps}\ln\frac{1}{\conf}.
    \end{equation}
    Furthermore, for $N_{\mathrm{S}}$ iid uncertainty samples $\bm{w}^{(i)}$, $i \in \mathbb{N}_1^{N_{\mathrm{S}}}$, let $\bm{\sigma} \coloneqq \col{\sigma\left(\bm{w}^{(1)}\right),\,\dots,\,\sigma\left(\bm{w}^{(N_{\mathrm{S}})}\right)}$ be the vector of scaling factors. 
    Then, $\Pr{\mathbb{Z}^{\mathrm{S}}(\sigma^*) \subseteq \mathbb{Z}^{\mathrm{P}}} \ge 1-\conf$ holds if $\sigma^*$ is the $\lceil \frac{\eps N_{\mathrm{S}}}{2} \rceil$-th smallest entry of $\bm{\sigma}$.
\end{proposition}

The complexity of the inner-approximation $\mathbb{Z}^{\mathrm{S}}(\sigma^*)$ is fully determined by the shape $\mathbb{Z}^{\mathrm{SAS}}$. 
To apply Proposition~\ref{prop:scaling}, the optimization in \eqref{eq:scaling} needs to be solved for $N_{\mathrm{S}}$ samples, which can be parallelized. For polytopic SASs, this optimization can be solved efficiently via linear programming.

\section{Adaptive Stochastic Predictive Control} \label{sec:method}
In order to control system~\eqref{eq:system_minimal} without access to state measurements, we consider the equivalent representation \eqref{eq:system_arx_nonminimal} based on the artifical extended state $\gls{xi}_{k}$ which is defined by past inputs and outputs, see \eqref{eq:extstate}.
Together with the unknown system parameters $\bm{\Phi}$, $\bm{\Psi}$, the unknown future disturbance realizations $\gls{d}_{\mathrm{f},k} \coloneqq \col{\gls{d}_{k},\,\dots,\,\gls{d}_{k+\gls{Tf}-1}}$ form the uncertainty $\bm{w} \coloneqq \left\{\bm{\Phi},\,\bm{\Psi},\,\gls{d}_{\mathrm{f},k}\right\}$ which needs to be taken into account when formulating the OCP of the proposed predictive controller.
As common in robust and stochastic predictive control, we parameterize the control input as
\begin{equation} \label{eq:inputdecomp}
    \gls{u}_k = \gls{v}_k + \bm{K} \gls{xi}_k,
\end{equation}
where $\gls{v}_k$ is the input correction term determined by the controller, and $\bm{K}$ is a stabilizing extended state feedback gain. Such feedback gains can be directly computed from the available input-output data by solving data-dependent linear matrix inequalities, e.g., see \cite{van2023behavioral}.
The conceptual OCP associated with the proposed predictive controller then reads
\begin{subequations} \label{eq:ocp_original}
		\begin{align}
		& \underset{\gls{v}_{\mathrm{f},k}}{\mathrm{minimize}} ~~~ J\left(\gls{xi}_k,\, \gls{v}_{\mathrm{f},k}\right) \label{eq:ocp_original_cost}\\		
		&\mathrm{s.t.}~\forall \,l \in \mathbb{N}_0^{\gls{Tf}-1}:\notag\\
        & ~~~~\hspace{2pt}\gls{y}_{l|k} = \bm{M}_l^{y}\left(\bm{w}\right) \, \col{\gls{xi}_k,\, \gls{v}_{\mathrm{f},k}}  + \bm{m}_l^{y}\left(\bm{w}\right), \label{eq:samplepred_output}\\
        & ~~~~\,\gls{u}_{l|k} = \bm{M}_l^{u}\left(\bm{w}\right) \, \col{\gls{xi}_k,\, \gls{v}_{\mathrm{f},k}} + \bm{m}_l^{u}\left(\bm{w}\right), \label{eq:samplepred_input}\\
		& \Pr{\gls{y}_{l|k} \in \mathbb{Y}} \ge 1-\eps, \label{eq:ocp_original_chancecons} \\
		& ~~~~\,\gls{u}_{l|k} \in \mathbb{U},\label{eq:ocp_original_inputcons}
		\end{align}
\end{subequations}
where $\gls{v}_{\mathrm{f},k} \coloneqq \col{\gls{v}_{0|k},\,\dots,\,\gls{v}_{\gls{Tf}-1|k}}$ is the vector of predicted input correction terms and $\bm{M}_l^{y}$, $\bm{m}_l^{y}$, $\bm{M}_l^{u}$, $\bm{m}_l^{u}$ are predictor parameters defined in the Appendix.

The OCP~\eqref{eq:ocp_original} is intractable due to the probabilistic chance constraint \eqref{eq:ocp_original_chancecons} and the uncertainty $\bm{w}=\left\{\bm{\Phi},\,\bm{\Psi},\,\gls{d}_{\mathrm{f},k}\right\}$.
We address this problem by proposing a sampling-based reformulation of the OCP~\eqref{eq:ocp_original}. This reformulation is combined with an adaptive set-membership approach that reduces uncertainty in $\bm{\Phi}$, $\bm{\Psi}$ online. Since the disturbance distribution is known by Assumption~\ref{assum:distset}, sampling future disturbance realizations $\gls{d}_{\mathrm{f},k}$ is straightforward. In the next section, we show how recorded input-output data map the disturbance distribution to a distribution over consistent system parameters from which we can then draw samples of $\bm{\Phi},\,\bm{\Psi}$.

\subsection{System Data and Adaptive Set of System Matrices}
The initially available input-output data trajectory needs to satisfy the following assumption.
\begin{assumption}[Initial data]\label{assum:trajData}
    An input-output trajectory 
    generated by system \eqref{eq:system_minimal} of length $T+\gls{Tini}$ is available
    and yields data matrices $\bm{H}_u \coloneqq \mat{\gls{ud}_1 & \cdots & \gls{ud}_T}$, $\bm{H}_y \coloneqq \mat{\gls{yd}_1 & \cdots & \gls{yd}_T}$, and $\bm{H}_{\xi} \coloneqq \mat{\gls{xid}_1 & \cdots & \gls{xid}_T}$ where $\bm{H}_{\xi,u} \coloneqq \col{\bm{H}_{\xi},\, \bm{H}_u}$ has full row-rank. 
\end{assumption}
\begin{remark}
    For sufficiently large disturbance levels, full row-rank of $\bm{H}_{\xi,u}$ is not restrictive since appropriate inputs $\bm{H}_u$ can be chosen
    and the extended state data $\bm{H}_{\xi}$ are perturbed by iid random disturbances.
    For small disturbance levels and thus nearly singular $\bm{H}_{\xi,u}$, full row-rank can be guaranteed by appropriate design of an alternative extended state \cite{alsalti2023notes}.
\end{remark}
Since the data originate from system~\eqref{eq:system_minimal}, they satisfy
\begin{equation} \label{eq:dynamics_data}
    \bm{H}_{y} = \bm{\Phi} \bm{H}_{\xi} + \bm{\Psi} \bm{H}_{u} + \bm{H}_{d},
\end{equation}
with unknown system parameters $\bm{\Phi}$, $\bm{\Psi}$ and unknown disturbance data matrix $\bm{H}_{d}\coloneqq \mat{\gls{dd}_1 & \cdots & \gls{dd}_T}$. 
By Assumption~\ref{assum:distset}, every column of the unknown data matrix $\bm{H}_{d}$ lies within the bounds~\eqref{eq:distset}, i.e., $\bm{H}_{d}$ satisfies $\bm{G}_d \bm{H}_{d} \le \bm{1}_T^{\top} \otimes \bm{g}_d$.
Exploiting~\eqref{eq:dynamics_data} yields a set $\mathbb{A}_0$ of system parameters consistent with the given data (including the true system parameters): 
\begin{equation} \label{eq:setsysmat_initial}
    \mathbb{A}_0 \coloneqq \left\{\mat{\bm{\Phi} & \bm{\Psi}} ~\left|~ -\bm{G}_{d} \mat{\bm{\Phi} & \bm{\Psi}} \bm{H}_{\xi,u} \le \bm{G}_{\mathbb{A}} \right.\right\}.
\end{equation}
Here, the right-hand-side of the inequality depends on the disturbance bounds and the output data as $\bm{G}_{\mathbb{A}} \coloneqq \bm{1}_T^{\top} \otimes \bm{g}_d - \bm{G}_d \bm{H}_{y}$.
Since $\bm{H}_{\xi,u}$ has full row-rank, $\mathbb{A}_0$ inherits compactness from the disturbance bounds~\eqref{eq:distset}. 

Additional data will further constrain the set of system parameters. 
Starting from the initial set $\mathbb{A}_0$, we employ a set membership update using newly retrieved data during online operation of the proposed control scheme. That is, the system parameter set $\mathbb{A}_{k+1}$ at the next time step is retrieved by intersection of the current set $\mathbb{A}_{k}$ with the set $\Delta_k$ of parameters consistent with the newly observed data $(\gls{u}_k,\gls{y}_k)$ at time step $k \ge 0$, i.e., $\mathbb{A}_{k+1} = \mathbb{A}_{k} \cap \Delta_k$ with
\begin{equation} \label{eq:setsysmat_newdata}
    \Delta_k \coloneqq \left\{\mat{\bm{\Phi} & \bm{\Psi}} \,\left|\, -\bm{G}_{d} \mat{\bm{\Phi} & \bm{\Psi}} \mat{\gls{xi}_k\\\gls{u}_k} \le \bm{g}_d - \bm{G}_d \gls{y}_k \right.\right\}.
\end{equation}
Although \eqref{eq:setsysmat_newdata} is generally unbounded, $\mathbb{A}_k$ is bounded due to the properties of the initial set \eqref{eq:setsysmat_initial}.
Furthermore, since $\mathbb{A}_k$ is a compact polytopic set, it can be described as a convex hull over its $N_{\mathrm{v},k}$ vertices $\mat{\bm{\Phi}_{k,j} & \bm{\Psi}_{k,j}}$, i.e.,
\begin{equation} \label{eq:setsysmat_vert}
    \mathbb{A}_k \coloneqq \conv{\left\{\mat{\bm{\Phi}_{k,j} & \bm{\Psi}_{k,j}}\right\}_{j=1}^{N_{\mathrm{v},k}} }.
\end{equation}

At each time step, $\mathbb{A}_k$ represents the support set of the probability distribution over system parameters. 
In order to derive this probability distribution over $\mathbb{A}_k$, we introduce an affine map from the disturbance data to the system matrices. 
Since $\bm{H}_{\xi,u}$ has full row-rank, there exists a selection matrix $\bm{\Omega} \in \mathbb{R}^{T \times (\glsd{xi} + \glsd{u})}$ such that $\tilde{\bm{H}}_{\xi,u} \coloneqq \bm{H}_{\xi,u}\bm{\Omega}$ is square and invertible. The matrix $\bm{\Omega}$ selects $\glsd{xi} + \glsd{u}$ columns of $\bm{H}_{\xi,u}$ that yield full rank. 
By \cite[Proposition~2]{teutsch2024sampling}, the unknown disturbance data matrix $\bm{H}_{d}$ can be fully parameterized by $\glsd{xi} + \glsd{u}$ of its columns, i.e., $\tilde{\bm{H}}_{d} \coloneqq \bm{H}_{d}\bm{\Omega}$, and the available data. Consequently, an invertible mapping from $\tilde{\bm{H}}_{d}$ to the unknown system parameters $\mat{\bm{\Phi} & \bm{\Psi}}$ is obtained as 
\begin{equation} \label{eq:sysmat_data_affinetranform}
    \mat{\bm{\Phi} & \bm{\Psi}} = \bm{H}_{y}\bm{\Omega}\tilde{\bm{H}}_{\xi,u}^{-1} - \tilde{\bm{H}}_{d} \tilde{\bm{H}}_{\xi,u}^{-1}. 
\end{equation}
Now, let $\tilde{\bm{f}}_{d}$ denote the extension of the disturbance distribution $\bm{f}_{d}$ from Assumption~\ref{assum:distset} to the distribution of the matrix of $\glsd{xi}+\glsd{u}$ disturbances $\tilde{\bm{H}}_{d}$, where the support of $\tilde{\bm{H}}_{d}$ considers the data-dependent bounds $\mathbb{A}_k$ via \eqref{eq:sysmat_data_affinetranform}, cf.~\cite{teutsch2024sampling}. The transformation \eqref{eq:sysmat_data_affinetranform} then yields the corresponding distribution over the system matrices as
\begin{equation} \label{eq:sysmat_distribution}
    \bm{f}_{\mathbb{A}}\left(\mat{\bm{\Phi} & \bm{\Psi}}\right) \coloneqq \lvert \det{\tilde{\bm{H}}_{\xi,u}}\rvert \,\tilde{\bm{f}}_{d}\left(\bm{H}_{y}\bm{\Omega} - \mat{\bm{\Phi} & \bm{\Psi}}\tilde{\bm{H}}_{\xi,u}\right).
\end{equation}
Note that the distribution $\bm{f}_{\mathbb{A}}$ is adapted online only via the adaptive support set $\mathbb{A}_k$; updating the parameterization \eqref{eq:sysmat_data_affinetranform} with new data is not required. Distribution~\eqref{eq:sysmat_distribution} allows for drawing samples of the uncertain system parameters $\bm{\Phi}$, $\bm{\Psi}$ within the bounds~$\mathbb{A}_k$ (e.g., via rejection sampling~\cite{martino2010generalized}). Such samples are used in the following to approximate the chance constraints~\eqref{eq:ocp_original_chancecons} at every time step $k$.

\subsection{Constraint Sampling} \label{sec:conssampling}
In order to render OCP \eqref{eq:ocp_original} tractable, we subsequently describe how to deterministically approximate the chance constraint \eqref{eq:ocp_original_chancecons} to a user-specified level of confidence via sampling-based probabilistic scaling (see Section~\ref{sec:sampling}). In our setting, the extended state $\gls{xi}_k$ and the sequence of future correction inputs $\gls{v}_{\mathrm{f},k}$ from \eqref{eq:ocp_original} will act as the deterministic decision variables for the approximation (as $\bm{\zeta}$ in Sec.~\ref{sec:sampling}), whereas the system parameters $\bm{\Phi}$, $\bm{\Psi}$ and future disturbances $\gls{d}_{\mathrm{f},k}$ take the role of the uncertainty $\bm{w}$.

Let us define the $\eps$-CCS for the predicted outputs $\gls{y}_{l|k}$ at time step $k$ for all $l\in\mathbb{N}_0^{\gls{Tf}-1}$ in the prediction horizon as
\begin{equation} \label{eq:ocpchancecons_CSS}
    \mathbb{Y}^{\mathrm{P}}_{k,l} = \left\{ \left.\vc{\gls{xi}_k \\ \gls{v}_{\mathrm{f},k}} ~\right|~ \Pr{\gls{y}_{l|k} \in \mathbb{Y}} \ge 1-\eps \right\}.
\end{equation}
Note that $\gls{y}_{l|k}$ is related to $\col{\gls{xi}_k,\, \gls{v}_{\mathrm{f},k}}$ via the predictor \eqref{eq:samplepred_output}.
The goal of the sampling-based approach is to determine a deterministic inner-approximation $\tilde{\mathbb{Y}}_{k,l}$ of the $\eps$-CCS~\eqref{eq:ocpchancecons_CSS} by using $N_{\mathrm{S}}$ iid samples $\bm{w}^{(i)} \coloneqq \left\{\bm{\Phi}^{(i)},\,\bm{\Psi}^{(i)},\,\gls{d}^{(i)}_{\mathrm{f},k}\right\}$, $i \in \mathbb{N}_1^{N_{\mathrm{S}}}$, of the uncertainty. More formally, for every time step $k$ the inner-approximation $\tilde{\mathbb{Y}}_{k,l}$ must satisfy $\Pr{\tilde{\mathbb{Y}}_{k,l} \subseteq \mathbb{Y}^{\mathrm{P}}_{k,l}} \ge 1-\conf$ with confidence $\conf$.

As discussed in Section~\ref{sec:sampling}, the probabilistic scaling approach following Proposition~\ref{prop:scaling} allows for a deterministic approximation of the $\eps$-CCS in \eqref{eq:ocpchancecons_CSS} with a pre-defined level of confidence $\beta$ using a scalable SAS $\mathbb{Y}^{\mathrm{S}}_{l}(\sigma_{k,l}) \coloneqq \{\bm{c}_l\} \oplus \sigma_{k,l} \left(\mathbb{Y}^{\mathrm{SAS}}_{l} \ominus \left\{\bm{c}_l\right\}\right)$, $l \in \mathbb{N}_0^{\gls{Tf}-1}$, where $\bm{c}_l$ is a center (e.g., Chebyshev or geometric center) of the SAS shape $\mathbb{Y}^{\mathrm{SAS}}_{l}$. 
A natural candidate for the SAS shape $\mathbb{Y}^{\mathrm{SAS}}_l$ is obtained by sampling $\tilde{N}_{\mathrm{S}}$ uncertainties $\tilde{\bm{w}}^{(i)}$ for design and constructing the polytope $\mathbb{Y}^{\mathrm{SAS}}_l \coloneqq \left\{  \col{\gls{xi}_k,\, \gls{v}_{\mathrm{f},k}} \,\left|\,  \tilde{\bm{G}}^y_l \col{\gls{xi}_k,\, \gls{v}_{\mathrm{f},k}} \le \bm{1}_{\tilde{N}_{\mathrm{S}}} \otimes \bm{g}_y - \tilde{\bm{g}}^y_l\right. \right\}$, with 
\begin{align*}
    \tilde{\bm{G}}^y_l &= \col{\bm{G}_{y} \bm{M}_l^y\left(\tilde{\bm{w}}^{(1)}\right),\,\ldots,\,\bm{G}_{y} \bm{M}_l^y\left(\tilde{\bm{w}}^{(\tilde{N}_{\mathrm{S}})}\right)}, \\
    \tilde{\bm{g}}^y_l &= \col{\bm{G}_{y} \bm{m}_l^y\left(\tilde{\bm{w}}^{(1)}\right),\,\ldots,\, \bm{G}_{y} \bm{m}_l^y\left(\tilde{\bm{w}}^{(\tilde{N}_{\mathrm{S}})}\right)},
\end{align*}
using the output constraint parameters $\bm{G}_{y}$, $\bm{g}_{y}$ from \eqref{eq:outputcons}.

The $\eps$-CSS approximation performed at every time step $k$ is summarized in Algorithm~\ref{alg:scaling}. After the SAS $\mathbb{Y}^{\mathrm{SAS}}_l$ is determined offline before the control phase, the approximating set $\tilde{\mathbb{Y}}_{k,l} \coloneqq \mathbb{Y}^{\mathrm{S}}_{l}(\sigma_{k,l}^*)$ is adapted online by redetermining the probabilistic scaling factors $\sigma_{k,l}^*$ based on sampled system parameters from \eqref{eq:sysmat_distribution} with updated support $\mathbb{A}_k$.
As samples are drawn at random, it might occur that $\sigma_{k}^*$ results smaller than $\sigma_{k-1}^*$, which would lead to contracting constraint sets and introduce unnecessary conservatism. As a remedy, we require $\sigma_{k,l}^* \ge \sigma_{k-1,l}^*$ and thereby guarantee that constraints can only relax, i.e., $\tilde{\mathbb{Y}}_{k-1,l} \subseteq \tilde{\mathbb{Y}}_{k,l}$.
\begin{algorithm}
\caption{$\eps$-CSS Approximation (cf. Proposition~\ref{prop:scaling})} \label{alg:scaling}
\begin{algorithmic}[1]
\REQUIRE Scalable SAS $\mathbb{Y}^{\mathrm{S}}_{l}(\sigma_{k,l})$ for $l \in \mathbb{N}_0^{\gls{Tf}-1}$, risk parameter $\eps \in (0,\,1)$, confidence $\conf \in (0,\,1)$.
\STATE Draw $N_{\mathrm{S}} \ge N_{\mathrm{PS}}(\eps,\,\conf)$ iid uncertainty samples $\bm{w}^{(i)}$, $i \in \mathbb{N}_1^{N_{\mathrm{S}}}$, considering the system parameter bounds $\mathbb{A}_k$.
\FOR{$l \in \mathbb{N}_1^{\gls{Tf}-1}$}
\STATE Determine the scaling factors $\sigma\left(\bm{w}^{(i)}\right)$, $i \in \mathbb{N}_1^{N_{\mathrm{S}}}$ via Definition~\ref{def:scaling}.
\STATE Set $\sigma_{k,l}^*$ as the $\lceil \frac{\eps N_{\mathrm{S}}}{2} \rceil$-th smallest scaling factor.
\STATE If $k \ge 1$ and $\sigma_{k,l}^* < \sigma_{k-1,l}^*$: set $\sigma_{k,l}^* \leftarrow \sigma_{k-1,l}^*$.
\STATE Construct the approximating set $\tilde{\mathbb{Y}}_{k,l} \coloneqq \mathbb{Y}^{\mathrm{S}}_{l}(\sigma_{k,l}^*)$.
\ENDFOR
\end{algorithmic}
\end{algorithm}

Due to the input parameterization \eqref{eq:inputdecomp}, uncertainty is also introduced into the predicted inputs $\gls{u}_{l|k}$ for $l\in\mathbb{N}_1^{\gls{Tf}-1}$. In order to accommodate this uncertainty, we approximate the hard input constraints \eqref{eq:ocp_original_inputcons} in the prediction horizon analogously to the output chance constraints by defining a new risk parameter $\eps_u$ and confidence level $\conf_u$, yielding polytopic constraint sets $\tilde{\mathbb{U}}_{k,l}$ for $l\in\mathbb{N}_0^{\gls{Tf}-1}$, where $\tilde{\mathbb{U}}_{k,0}$ is such that the hard constraints $\mathbb{U}$ hold for the input that is actually applied, i.e., $\gls{v}_{0|k} + \bm{K} \gls{xi}_k \in \mathbb{U}$. Analogously to the output constraint sampling, this approximation is accomplished using the input predictor \eqref{eq:samplepred_input}.
In order to conclude the constraint sampling, we define the aggregate constraint set for time step $k$ as the intersection of the sampled input and output constraints $\tilde{\mathbb{U}}_{k,l}$, $\tilde{\mathbb{Y}}_{k,l}$ for all predicted steps $l\in\mathbb{N}_0^{\gls{Tf}-1}$, i.e.,
\begin{equation} \label{eq:constraints_red}
    \mathbb{C}_k \coloneqq \bigcap_{l=0}^{\gls{Tf}-1} \tilde{\mathbb{U}}_{k,l} \cap \tilde{\mathbb{Y}}_{k,l}
\end{equation}
Note that $\mathbb{C}_{k} \subseteq \mathbb{C}_{k+1}$ by design (see Algorithm~\ref{alg:scaling}, Step~5).

To render the control scheme recursively feasible, we construct an additional constraint $\mathbb{C}^{\mathrm{R}}_k$ on the first predicted step as proposed by \cite{lorenzen2017stochastic}; we refer to \cite{kerz2023datadriven} for a more~detailed discussion. Let $\mathbb{C}_{\gls{Tf},k}$ denote the set of feasible initial extended states and first inputs, obtained by projection of~\eqref{eq:constraints_red}:
\begin{equation} \label{eq:constraints_red_projected}
    \mathbb{C}_{\gls{Tf},k} \coloneqq \left\{ \vc{\gls{xi}_k \\ \gls{v}_{0|k}} \left|~ \begin{array}{c}
         \exists \, \gls{v}_{1|k},\, \dots,\, \gls{v}_{\gls{Tf}-1|k} \in \mathbb{R}^{\glsd{u}}:\\
         \col{\gls{xi}_k,\,\gls{v}_{\mathrm{f},k}} \in \mathbb{C}_k
    \end{array}\right. \right\}.
\end{equation}
Based on the matrix vertices $\tilde{\bm{A}}_{\mathrm{cl},k,j} \coloneqq \tilde{\bm{A}}_{k,j} + \tilde{\bm{B}}_{k,j}\bm{K}$, $j \in \mathbb{N}_1^{N_{\mathrm{v},k}}$ constructed from \eqref{eq:setsysmat_vert} and \eqref{eq:system_arx_nonminimal_helper}, and the disturbance bound $\gls{dset}$ from \eqref{eq:distset}, we determine a robust control invariant (RCI) set for system \eqref{eq:system_arx_nonminimal} with $\col{\gls{xi}_k,\,\gls{v}_{0|k}} \in \mathbb{C}_{\gls{Tf},k}$ of the form $\mathbb{C}_{\xi,k}^{\infty} \coloneqq \left\{ \gls{xi} \in \mathbb{R}^{\glsd{xi}} \left|~ \bm{G}^{\infty}_{\xi,k} \gls{xi} \le \bm{g}^{\infty}_{\xi,k} \right.\right\}$. For appropriate algorithms to determine $\mathbb{C}_{\xi,k}^{\infty}$, we refer to \cite[Section~5.3]{blanchini2015set}. Finally, the first-step constraint set is constructed as
\begin{equation} \label{eq:constraints_firststep}
    \mathbb{C}^{\mathrm{R}}_k \coloneqq \left\{ \hspace{-1mm}\vc{\gls{xi}_k \\ \gls{v}_{\mathrm{f},k}} \left| \hspace{-1mm} \begin{array}{l} 
	\forall \,\gls{d} \in \gls{dset},\, j \in \mathbb{N}_1^{N_\mathrm{v},k}: \\
    \tilde{\bm{A}}_{\mathrm{cl},k,j} \gls{xi}_k +\tilde{\bm{B}}_{k,j} \gls{v}_{0|k} + \tilde{\bm{E}} \gls{d} \in \mathbb{C}_{\xi,k}^{\infty}
	\end{array} \hspace{-3mm}\right. \right\}.
\end{equation}

The constraint set for the predictive controller at time step $k$ is thus given as $\mathbb{C}_k \cap \mathbb{C}^{\mathrm{R}}_k$. For initial feasibility, this intersection must be non-empty, i.e., the bounds in Assumptions~\ref{assum:distset} must be suitably tight, as large disturbance bounds likely lead to an empty intersection $\mathbb{C}_k$ of the sampled constraint sets. 

\subsection{Reformulation of the Cost Function}
Besides the constraints, also the cost function \eqref{eq:cost_expected} needs to be deterministically approximated. We do so by sample average approximation~\cite{kim2015guide} based on the predictors \eqref{eq:samplepred_output},\,\eqref{eq:samplepred_input}:
Given an uncertainty sample $\bm{w}^{(i)}$ and the associated input and output predictions $\gls{u}^{(i)}_{\mathrm{f},k}$, $\gls{y}^{(i)}_{\mathrm{f},k}$, the corresponding cost is
\begin{equation}
    J^{(i)} = \norm{\gls{y}^{(i)}_{\mathrm{f},k} - \gls{yrefvec}_{k}}^2_{\tilde{\bm{Q}}} + \norm{\gls{u}^{(i)}_{\mathrm{f},k}-\gls{urefvec}_{k}}^2_{\tilde{\bm{R}}},
\end{equation}
where $\gls{yrefvec}_{k} := \col{\gls{yref}_k,\,\dots,\,\gls{yref}_{k+\gls{Tf}-1}}$, $\gls{urefvec}_{k} := \col{\gls{uref}_k,\,\dots,\,\gls{uref}_{k+\gls{Tf}-1}}$, $\tilde{\bm{Q}} \coloneqq \bm{I}_{\gls{Tf}} \otimes \bm{Q}$, and $\tilde{\bm{R}} \coloneqq \bm{I}_{\gls{Tf}} \otimes \bm{R}$. By expressing predicted inputs and outputs 
in terms of the deterministic initial extended state $\gls{xi}_k$ and predicted correction inputs $\gls{v}_{\mathrm{f},k}$ via the predictors \eqref{eq:samplepred_output}, \eqref{eq:samplepred_input}, the sampled cost can be reformulated as
\begin{equation} \label{eq:cost_reform}
    J^{(i)} = \norm{\vc{\gls{xi}_k \\ \gls{v}_{\mathrm{f},k}}}^2_{\bm{Q}^{(i)}_{\mathrm{S}}} + 2{\bm{q}^{(i)}_{\mathrm{S},k}}^{\top} \vc{\gls{xi}_k \\ \gls{v}_{\mathrm{f},k}} + c^{(i)}_{\mathrm{S}},
\end{equation}
with the sampled parameters $\bm{Q}^{(i)}_{\mathrm{S}} \coloneqq \bm{Q}_{\mathrm{S}}\left(\bm{w}^{(i)}\right)$, $\bm{q}^{(i)}_{\mathrm{S},k} \coloneqq \bm{q}_{\mathrm{S},k}\left(\bm{w}^{(i)}\right)$, and $c^{(i)}_{\mathrm{S}} \coloneqq c_{\mathrm{S}}\left(\bm{w}^{(i)}\right)$ defined in the Appendix.

Based on $N_{\mathrm{avg}}$ uncertainty samples $\bm{w}^{(i)}$, the sample-average cost function that approximates \eqref{eq:cost_expected} then results in
\begin{equation} \label{eq:cost_sampleaverage}
    \hat{J}\left(\gls{xi}_k,\,\gls{v}_{\mathrm{f},k}\right) = \norm{\vc{\gls{xi}_k \\ \gls{v}_{\mathrm{f},k}}}^2_{\hat{\bm{Q}}_{\mathrm{S}}} + 2\hat{\bm{q}}^{\top}_{\mathrm{S},k} \vc{\gls{xi}_k \\ \gls{v}_{\mathrm{f},k}} + \hat{c}_{\mathrm{S}},
\end{equation}
where $\hat{c}_{\mathrm{S}} \coloneqq \left(1/{N_{\mathrm{avg}}}\right)\sum_{i=1}^{N_{\mathrm{avg}}} c^{(i)}_{\mathrm{S}}$ is a constant term that can therefore be neglected in the optimization, and 
the weights $\hat{\bm{Q}}_{\mathrm{S}} \coloneqq \left(1/{N_{\mathrm{avg}}}\right)\sum_{i=1}^{N_{\mathrm{avg}}} \bm{Q}^{(i)}_{\mathrm{S}}$, $
\hat{\bm{q}}_{\mathrm{S},k} \coloneqq \left(1/{N_{\mathrm{avg}}}\right)\sum_{i=1}^{N_{\mathrm{avg}}} \bm{q}^{(i)}_{\mathrm{S},k}.$

\subsection{Control Algorithm}
Based on the reformulated constraints \eqref{eq:constraints_red}, \eqref{eq:constraints_firststep} and cost \eqref{eq:cost_sampleaverage}, the OCP of the proposed predictive controller is
\begin{subequations} \label{eq:ocp}
	\begin{align}
		\underset{\gls{v}_{\mathrm{f},k}}{\mathrm{minimize}} ~~~&   \hat{J}\left(\gls{xi}_k,\, \gls{v}_{\mathrm{f},k}\right) \\		
		\mathrm{s.t. }~~~ & \col{\gls{xi}_k,\, \gls{v}_{\mathrm{f},k}} \in \mathbb{C}_k \cap \mathbb{C}_k^{\mathrm{R}}, \label{eq:ocp_constraints}
	\end{align}
\end{subequations}
where the constraint sets $\mathbb{C}_k$, $\mathbb{C}_k^{\mathrm{R}}$ are adapted over time $k$ based on the rescaled SASs as in Algorithm~\ref{alg:scaling}.
The implicit control law associated with OCP~\eqref{eq:ocp} reads as $\bm{\kappa}\left(\gls{xi}_k\right):= \gls{u}^{*}_k = \gls{v}^{*}_{0|k} + \bm{K} \gls{xi}_k$, where $\gls{v}^{*}_{0|k}$ is the first input of the optimal input vector $\gls{v}^*_{\mathrm{f},k}$.
The overall algorithm of the proposed controller is summarized in Algorithm~\ref{alg:controller}, split into an offline (before control) and online (during control) phase.

\begin{algorithm}
\caption{Adaptive Stochastic Predictive Control} \label{alg:controller}
\begin{algorithmic}[1]
\renewcommand{\algorithmicensure}{\textbf{Offline Phase:}}
\ENSURE
\STATE Retrieve an initial input-output data trajectory from system \eqref{eq:system_minimal} that satisfies Assumption~\ref{assum:trajData}.
\STATE Compute the initial set of system matrices $\mathbb{A}_0$ \eqref{eq:setsysmat_initial}.
\STATE Determine a stabilizing feedback gain $\bm{K}$.
\STATE Construct suitable SASs for the constraint sets and compute the initial constraint set $\mathbb{C}_0$~\eqref{eq:constraints_red} via Algorithm~\ref{alg:scaling}.
\STATE Determine the initial RCI set $\mathbb{C}_{\xi,0}^{\infty}$ and compute the corresponding first-step constraint $\mathbb{C}^{\mathrm{R}}_0$~\eqref{eq:constraints_firststep}.
\STATE Determine the weights of the cost function \eqref{eq:cost_sampleaverage}.
\renewcommand{\algorithmicensure}{\textbf{Online Phase:}}
\ENSURE for all $k \ge 0$:
\STATE Construct the current extended state $\gls{xi}_k$ from the most recent past $\gls{Tini}$ input-output measurements.
\STATE Solve the OCP \eqref{eq:ocp} to retrieve $\gls{v}^{\ast}_{0|k}$.
\STATE Apply the input $\gls{u}^*_{k} = \gls{v}^{\ast}_{0|k} + \bm{K} \gls{xi}_k$ to the system.
\STATE Adapt the set of system matrices $\mathbb{A}_{k+1}$ via \eqref{eq:setsysmat_newdata}.
\STATE Obtain $\tilde{\mathbb{Y}}_{k+1,l}$, $\tilde{\mathbb{U}}_{k+1,l}$, $l \in \mathbb{N}_0^{\gls{Tf}-1}$ using Algorithm \ref{alg:scaling}.
\STATE Compute adapted constraint sets $\mathbb{C}_{k+1}$ \eqref{eq:constraints_red}, $\mathbb{C}_{k+1}^{\mathrm{R}}$ \eqref{eq:constraints_firststep} based on $\mathbb{A}_{k+1}$, $\tilde{\mathbb{Y}}_{k+1,l}$, $\tilde{\mathbb{U}}_{k+1,l}$, $l \in \mathbb{N}_0^{\gls{Tf}-1}$.
\end{algorithmic}
\end{algorithm}

\subsection{Control-theoretic Properties} \label{sec:properties}
In the following, we present control theoretic properties of the proposed controller, enabled by the robust first-step constraint~\eqref{eq:constraints_firststep} as commonly used in the related literature, e.g., \cite{teutsch2024sampling,mammarella2018offline,lorenzen2017stochastic}. As a novel contribution, we show that the guarantees are preserved despite the adaptation scheme.

\begin{theorem}[Recursive Feasibility] \label{th:recfeas}
    Let $\mathbb{F}_k\left(\gls{xi}_{k}\right)$ be the set of all feasible input sequences for the OCP~\eqref{eq:ocp} at time step~$k$ and the corresponding extended state $\gls{xi}_{k}$, i.e.,
    \begin{equation} \label{eq:feasset}
        \mathbb{F}_k\left(\gls{xi}_{k}\right) = \left\{ \gls{v}_{\mathrm{f},k} ~\left|~ \col{\gls{xi}_k,\, \gls{v}_{\mathrm{f},k}} \in \mathbb{C}_k \cap \mathbb{C}^{\mathrm{R}}_k \right.\right\}.
    \end{equation}
    For every realization of $\gls{d}_{k} \in \gls{dset}$, it holds that
    $\mathbb{F}_k\left(\gls{xi}_{k}\right) \neq \emptyset \implies \mathbb{F}_{k+1}\left(\gls{xi}_{k+1}\right) \neq \emptyset$ under the proposed control law.
\end{theorem}

\begin{proof}
    By robustness of the first-step constraint \eqref{eq:constraints_firststep}, $\col{\gls{xi}_k,\, \gls{v}_{\mathrm{f},k}} \in \mathbb{C}^{\mathrm{R}}_k$ implies $\gls{xi}_{k+1} \in \mathbb{C}_{\xi,k}^{\infty}$. As $\mathbb{C}_{\xi,k}^{\infty} \subset \left\{\gls{xi} \,\left|\, \mathbb{F}_k\left(\gls{xi}\right) \neq \emptyset \right.\right\}$, $\mathbb{C}_{k} \subseteq \mathbb{C}_{k+1}$ , and $\mathbb{A}_{k+1} \subseteq \mathbb{A}_k$ hold by construction, we have $\mathbb{C}_{\gls{Tf},k} \subseteq \mathbb{C}_{\gls{Tf},k+1}$ and thus $\mathbb{C}_{\xi,k}^{\infty} \subseteq \mathbb{C}_{\xi,k+1}^{\infty} \subset \left\{\gls{xi} \,\left|\, \mathbb{F}_{k+1}\left(\gls{xi}\right) \neq \emptyset \right.\right\}$, concluding the proof.
\end{proof}

Note that, as $\mathbb{C}_{k} \subseteq \mathbb{C}_{k+1}$ and $\mathbb{C}_{\xi,k}^{\infty} \subseteq \mathbb{C}_{\xi,k+1}^{\infty}$ hold by design, recursive feasibility is preserved even if the constraint sets are not adapted at every time step.

\begin{corollary}[Closed-loop Constraint Satisfaction]
    For $\gls{xi}_0 \in \mathbb{C}_{\xi,0}^{\infty}$, the closed-loop system under the proposed control law satisfies the output chance constraint \eqref{eq:outputcons} with confidence $\conf$ and the input constraints \eqref{eq:inputcons} for all $k \ge 0$.
\end{corollary}

\begin{proof}
    For $\gls{xi}_0 \in \mathbb{C}_{\xi,0}^{\infty}$, a feasible pair $\col{\gls{xi}_0,\, \gls{v}_{\mathrm{f},0}} \in \mathbb{C}_0$ exists. Theorem~\ref{th:recfeas} and the constraint $\gls{v}_{0|k} + \bm{K} \gls{xi}_k \in \mathbb{U} \subset \mathbb{C}_k$ yield satisfaction of the input constraint \eqref{eq:inputcons} in closed-loop. 
    Furthermore, by design of the constraint sets, it holds that $\mathbb{C}_k \subseteq \tilde{\mathbb{Y}}_{k,0}$ and $\tilde{\mathbb{Y}}_{k,0} \subseteq \mathbb{Y}_{k,0}^{\mathrm{P}}$ with probability $1-\conf$. Thus, the chance constraint \eqref{eq:outputcons} is satisfied with confidence $\conf$ for all feasible $\col{\gls{xi}_k,\, \gls{v}_{\mathrm{f},k}} \in \mathbb{C}_k$, $k \ge 0$, which is sufficient for closed-loop chance constraint satisfaction.
\end{proof}

\subsection{Discussion on Computational Aspects} \label{sec:discuss}

As highlighted in Algorithm~\ref{alg:controller}, the proposed scheme is split into an offline and online phase. The offline phase consists of computationally heavy tasks such as the design of the candidate SAS sets that are used for chance constraint approximation, and the computation of the (maximal) RCI subset $\mathbb{C}_{\xi,0}^{\infty}$ for the first-step constraint~\eqref{eq:constraints_firststep}. 
For the online phase, the complexity of the OCP~\eqref{eq:ocp} --- a quadratic program --- is determined by the chosen SASs. The adaptation of the constraint sets~\eqref{eq:constraints_red} via probabilistic scaling consists of a series of linear programs, one for each step of the horizon if both the vertex and half-space representation of the SASs are known. 
Updates of the RCI subset $\mathbb{C}_{\xi,k}^{\infty}$ and the subsequent first-step constraint \eqref{eq:constraints_firststep} are typically more expensive. 
Alternatively to a full recomputation of the RCI subset, one can choose $\mathbb{C}_{\xi,k+1}^{\infty} \coloneqq \mu_k \mathbb{C}_{\xi,k}^{\infty}$ as $\mathbb{C}_{\xi,k}^{\infty}$ is a (non-maximal) RCI subset of $\mathbb{C}_{\gls{Tf},k+1}$ due to $\mathbb{C}_k \subseteq \mathbb{C}_{k+1}$, with scaling factor $\mu_k \ge 1$ such that $\mu_k \mathbb{C}_{\xi,k}^{\infty}$ is still contained in $\mathbb{C}_{\gls{Tf},k+1}$ \cite[Theorem~6]{rakovic2004invariant}. The largest possible $\mu_k$ can be determined via a series of linear programs.
In contrast, the SAS shapes cannot easily be updated without losing recursive feasibility, since containment is then no longer guaranteed solely by non-decreasing scaling factors.

In practice, adapting some of the above mentioned sets at each time step may likely be too costly. However, by design, all theoretical properties are still valid if updates occur in asynchronous fashion, i.e., at irregular user-specified intervals. For example, a separate machine or external platform may inject new constraints (based on the latest batch of data) into the control algorithm whenever processing is finished.

A direct adaptation of the set of system matrices $\mathbb{A}_k$ using~\eqref{eq:setsysmat_newdata} can lead to unbounded growth of set complexity in terms of linear inequalities. However, at the cost of additional conservatism, an over-approximation of $\mathbb{A}_k$ with fixed complexity can be employed \cite{lorenzen2019robust}. For the sets that are computed offline, redundant constraints should be removed \cite{lorenzen2017stochastic}.
In order to generate the required samples from arbitrary distributions supported by polytopic sets, one can employ rejection sampling \cite{martino2010generalized} or Hit-and-Run algorithms \cite{mete2012patternHR}.

\section{Numerical Evaluation} \label{sec:eval}
In this section, we evaluate the proposed adaptive control scheme on an example system in simulation.

\subsection{Simulation Setup}
The considered system represents a linearized DC-DC converter~\cite{lazar2008input}, which follows the dynamics from \eqref{eq:system_arx_nonminimal} with
\begin{equation}
    \bm{\Phi} = \mat{4.798 & 1 & 0.008 \\ 
                     0.115 & -0.143 & 0.996},~~
    \bm{\Psi} = \bm{0},
\end{equation}
and $\gls{Tini} = 1$.
The parameters of the polytopic input and output constraint sets~\eqref{eq:constraints} are defined as $\bm{G}_u = \col{\bm{I}_{\glsd{u}},-\bm{I}_{\glsd{u}}}$, $\bm{g}_u = 0.2\cdot\bm{1}_{2\glsd{u}}$, $\bm{G}_y = \col{\bm{I}_{\glsd{d}},-\bm{I}_{\glsd{d}}}$, $\bm{g}_y = 3\cdot\bm{1}_{2\glsd{y}}$.
The disturbance is bounded by the polytopic support set \eqref{eq:distset} with $\bm{G}_d = \col{\bm{I}_{\glsd{y}},-\bm{I}_{\glsd{y}}}$, $\bm{g}_d = \bm{1}_{2} \otimes \mat{0.1 & 0.05}^{\top}$, and is uniformly distributed within the bounds.

The initial input-output data trajectory (see Assumption~\ref{assum:trajData}) is retrieved by applying random admissible inputs for $T = 15$ time steps.
For the cost function~\eqref{eq:cost_expected}, we choose a prediction horizon of $\gls{Tf} = 5$ and the weighting matrices $\bm{R} = 1$, $\bm{Q} = \mat{\col{1,\,0} & \col{0,\,100}}$.
The feedback gain $\bm{K}$ is determined as $\bm{K} = \mat{-0.901 & -0.202 & 0.138}$ using the vertices of \eqref{eq:setsysmat_initial}, see \cite{teutsch2024sampling}. The candidate SASs are constructed by using $1\,000$ to $5\,000$ ``design" samples (see Sec.~\ref{sec:conssampling}); the total number of vertices and half-space inequalities of the SASs per prediction stage are listed in Table~\ref{tab:sas_comparison}.
For the constraint sampling (see Algorithm~\ref{alg:scaling}), we choose the risk parameter $\eps = 0.1$ and confidence level $\conf = 10^{-3}$.

The control goal is to track a square-wave output reference trajectory $\gls{yref}_{k}$, oscillating between $\col{0,\,-3}$ and $\col{0,\,3}$ every $50$ time steps.
The system is initialized at $\bm{\xi}_0 = \bm{0}$, and $T_{\mathrm{sim}} = 150$ control iterations are performed. We adapt the constraint set $\mathbb{C}_k$ via Algorithm~\ref{alg:scaling} in every time step. The RCI set $\mathbb{C}^{\infty}_{\xi,k}$ for the first-step constraint \eqref{eq:constraints_firststep} is rescaled at every $10$ time steps and fully recomputed only once at time step $k=80$ (see discussion in Section~\ref{sec:discuss}).

\subsection{Simulation Results}

We compare the proposed adaptive approach to the corresponding non-adaptive scheme with same initial data. Simulations are carried out in MATLAB on an AMD Ryzen 5 Pro 3500U with \texttt{quadprog} solving the OCP~\eqref{eq:ocp}.

\begin{figure}[!h]
    \centering
    \vspace{2mm}
    \includegraphics[clip, trim=2cm 20.5cm 10cm 1.85cm]{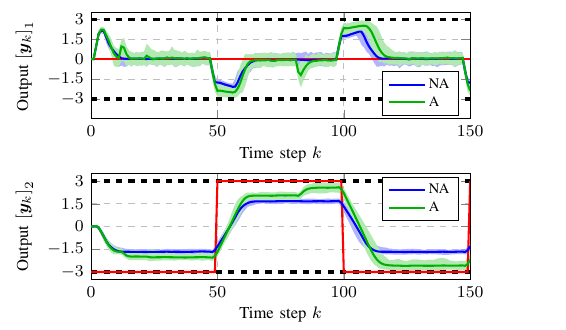}
    \caption{Closed-loop output trajectories for reference tracking. The output constraints are depicted as black dashed lines and the output reference $\gls{yref}_{k}$ is depicted in red. (NA: Non-adaptive scheme; ~A: Adaptive scheme)}
    \label{fig:output}
\end{figure}

\begin{figure}[!h]
    \centering
    \includegraphics[clip, trim=2cm 23.2cm 10cm 1.85cm]{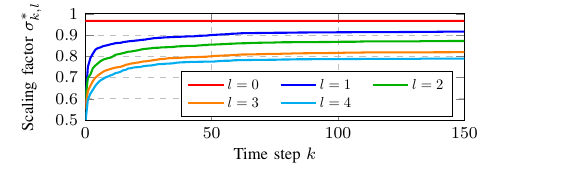}
    \caption{Mean evolution of the scaling factors $\sigma^*_{k,l}$ for output $\eps$-CSS approximation $\tilde{\mathbb{Y}}_{k,l}$ following Algorithm~\ref{alg:scaling}.}
    \label{fig:scaling}
\end{figure}

\begin{table}  
    \vspace*{2mm}
    \caption{SAS complexity and computation time of SAS scaling }
    \label{tab:sas_comparison}
    \begin{center}    
    \vspace*{-5mm}
        \begin{tabular}{c||c|c|c|c|c}
                Prediction stage $l$ & 0 & 1 & 2 & 3 & 4 \\ 
            \hhline{=||=|=|=|=|=}
                {\#vertices} & $248$ & $1\,144$ & $5\,088$ & $10\,980$ & $21\,814$\\
            \hline
                {\#half-spaces} & $74$ & $142$ & $240$ & $232$ & $208$\\
            \hline
                {Mean (in $\,\unit{\milli\s}$)} & $0.80$ & $1.90$ & $5.30$ & $14.92$ & $32.94$\\ 
            \hline 
                {Standard dev. (in $\,\unit{\milli\s}$)} & $0.56$ & $0.81$ & $1.32$ & $2.06$ & $3.66$\\ 
            \hline 
        \end{tabular}
    \end{center}
\end{table}

Fig.~\ref{fig:output} depicts the closed-loop output trajectories for $100$ runs subject to random disturbances. Compared to the non-adaptive scheme, the proposed adaptive scheme is able to drive the system outputs closer to the bounds after obtaining more data from the system, leading to an increased control performance over time. To depict the adaptation of constraint sets, Fig.~\ref{fig:scaling} shows the mean evolution of the scaling factors for the output $\eps$-CSS approximation as new data stream in. When comparing Fig.~\ref{fig:output} and Fig.~\ref{fig:scaling}, one can see that the performance of the proposed controller is not directly linked to the increase of the scaling factors; the adapted constraint sets enable an increase in control performance once the RCI set is adapted (e.g., rescaling at $k=10$ and recomputation at $k=80$). Table~\ref{tab:sas_comparison} provides the mean and standard deviation of computation times for rescaling the SASs of every prediction stage. The mean computation time for solving OCP~\eqref{eq:ocp}, rescaling the RCI set, and recomputing the RCI set is $1.82\,\unit{\milli\s}$, $426.78\,\unit{\milli\s}$, and $19.95\,\unit{\s}$, respectively.

At last, we compute the total tracking cost over the full simulation time for both methods as $J_{\mathrm{tot}} = \sum_{k=0}^{T_{\mathrm{sim}}} \left( \norm{\gls{y}_{k} - \gls{yref}_{k}}_{\bm{Q}}^2 + \norm{\gls{u}_{k}}_{\bm{R}}^2 \right)$.
The non-adaptive scheme results in an average total cost of $49\,594$, while the adaptive scheme results in an average total cost of $41\,512$, leading to an improvement of $16.30\,\%$ in average.

\section{Conclusion} \label{sec:conclusion}
We proposed an adaptive output-feedback predictive control scheme for constrained linear systems with unknown parameters subject to bounded probabilistic disturbances. 
Input-output data map the distributional knowledge of the disturbance to a distribution over the system parameters, with an adaptive support set based on set membership identification.
A probabilistic scaling approach based on samples from the set of system parameters and future disturbances allows for adaptive deterministic approximations of chance constraints, which can be efficiently performed online.  
In simulation, the proposed adaptive controller quickly improves beyond the equivalent controller solely based on the initial data. 
At the cost of rescaling or recomputing a robust control invariant set that enables recursive feasibility and closed-loop constraint satisfaction, the control performance is further improved.

Left open in this work is the analysis of possible convergence properties of the closed-loop system and convergence of the set of system parameters as in \cite{lu2023adaptivePE}.
Other interesting directions for future research are adaptation of SAS shapes while preserving recursive feasibility, and recursive feasibility guarantees without the use of computationally demanding robust control invariant sets.

\appendix
\section{Parameters}
For the parameters that are specified in the following, we omit the dependency on $\bm{w} = \left\{\bm{\Phi},\,\bm{\Psi},\,\gls{d}_{\mathrm{f},k}\right\}$ to simplify notation.
The parameters of the predictors \eqref{eq:samplepred_output},\,\eqref{eq:samplepred_input} are retrieved by explicitly solving the dynamics \eqref{eq:system_arx_nonminimal} with the input decomposition \eqref{eq:inputdecomp} for the predicted output $\gls{y}_{l|k}$, extended state $\gls{xi}_{l|k}$, and input $\gls{u}_{l|k}$, $l \in \mathbb{N}_0^{\gls{Tf}}$. For $l = 0$, we have $\bm{M}_0^{\xi} = \mat{\bm{I}_{\glsd{xi}} & \bm{0}_{\glsd{xi} \times \gls{Tf}\glsd{u}}}$, $\bm{m}_0^{\xi} = \bm{0}$, $\bm{M}_0^{u} = \mat{\bm{K} & \bm{I}_{\glsd{u}} & \bm{0}_{\glsd{u} \times (\glsd{Tf}-1)\glsd{u}}}$, $\bm{m}_0^{u} = \bm{0}$, $\bm{M}_0^{y} = \mat{\bm{\Phi}_{\mathrm{cl}} & \bm{\Psi} & \bm{0}_{\glsd{y} \times (\glsd{Tf}-1)\glsd{u}}}$, and $\bm{m}_0^{y} = \mat{\bm{I}_{\glsd{y}} & \bm{0}} \gls{d}_{\mathrm{f},k}$, with $\bm{\Phi}_{\mathrm{cl}} = \bm{\Phi} + \bm{\Psi}\bm{K}$. For $l \in \mathbb{N}_1^{\gls{Tf}}$, we obtain\\[-4mm]
\begin{subequations}\label{eq:predictors}
    \begin{align} 
        \bm{M}_{l}^{\xi} &= \mat{\tilde{\bm{A}}_{\mathrm{cl}}^l & \bm{\Omega}_{l}^v}, &\bm{m}_{l}^{\xi} &= \bm{\Omega}_{l}^d \gls{d}_{\mathrm{f},k},\\
        \bm{M}_l^y &= \bm{\Phi}_{\mathrm{cl}}\bm{M}_{l}^{\xi} + \mat{\bm{0} & \bm{\Psi} \bm{\Theta}_{l,\glsd{u}}}, &\bm{m}_l^y &= \bm{m}_{l}^{\xi} + \bm{\Theta}_{l,\glsd{y}} \gls{d}_{\mathrm{f},k}, \\
        \bm{M}_l^u &= \bm{K}\bm{M}_{l}^{\xi} + \mat{\bm{0} & \bm{\Theta}_{l,\glsd{u}}},  &\bm{m}_l^u &= \bm{K} \bm{m}_{l}^{\xi},    
    \end{align}
\end{subequations}
with $\tilde{\bm{A}}_{\mathrm{cl}} = \tilde{\bm{A}} + \tilde{\bm{B}}\bm{K}$ and
\begin{subequations}
    \begin{align}
        &\bm{\Omega}_{l}^v = \mat{\tilde{\bm{A}}_{\mathrm{cl}}^{l-1} \tilde{\bm{B}} & \tilde{\bm{A}}_{\mathrm{cl}}^{l-2} \tilde{\bm{B}} & \cdots & \tilde{\bm{B}} & \bm{0}_{\glsd{xi}\times(\gls{Tf}-l)\glsd{u}}}, \\
        &\bm{\Omega}_{l}^d = \mat{\tilde{\bm{A}}_{\mathrm{cl}}^{l-1} \tilde{\bm{E}} & \tilde{\bm{A}}_{\mathrm{cl}}^{l-2} \tilde{\bm{E}} & \cdots & \tilde{\bm{E}} & \bm{0}_{\glsd{xi}\times(\gls{Tf}-l)\glsd{y}}}, \\
        &\bm{\Theta}_{l,\glsd{u}} = \mat{\bm{0}_{\glsd{u}\times l \glsd{u}} & \bm{I}_{\glsd{u}} & \bm{0}_{\glsd{u}\times (\gls{Tf}-l-1)\glsd{u}}}.
    \end{align}
\end{subequations}
By concatenating the stage-wise predictors \eqref{eq:predictors}, we obtain the cost parameters for the reformulated cost function \eqref{eq:cost_reform}:
\begin{subequations}
    \begin{align}
        \bm{Q}_{\mathrm{S}} &= \bm{M}_y^{\top} \tilde{\bm{Q}} \bm{M}_y + \bm{M}_u^{\top} \tilde{\bm{R}} \bm{M}_u,\\ 
        \bm{q}_{\mathrm{S},k} &=  \bm{M}_y^{\top} \tilde{\bm{Q}} \gls{yrefvec}_k + \bm{M}_u^{\top} \tilde{\bm{R}} \gls{urefvec}_k,\\ 
        c_{\mathrm{S}} &= \bm{m}_y^{\top} \tilde{\bm{Q}} \bm{m}_y + \bm{m}_u^{\top} \tilde{\bm{R}} \bm{m}_u,
    \end{align}
\end{subequations}
with the multi-step predictor parameters $\bm{M}_{y} = \col{\bm{M}_{0}^{y},\,\dots,\,\bm{M}_{\gls{Tf}-1}^{y}}$, $\bm{m}_{y} = \col{\bm{m}_{0}^{y},\,\dots,\,\bm{m}_{\gls{Tf}-1}^{y}}$, $\bm{M}_{u} = \col{\bm{M}_{0}^{u},\,\dots,\,\bm{M}_{\gls{Tf}-1}^{u}}$, and $\bm{m}_{u} = \col{\bm{m}_{0}^{u},\,\dots,\,\bm{m}_{\gls{Tf}-1}^{u}}$. Note that for $\bm{q}_{\mathrm{S},k}$ we neglected terms involing $\gls{d}_{\mathrm{f},k}$ due to its i.i.d. and zero-mean properties.

\bibliographystyle{ieeetr}
\bibliography{mybib}

\end{document}